\documentclass{amsart}
\usepackage{amsmath}
\usepackage{amsthm}
\usepackage{amssymb}
\usepackage{amsfonts}
\usepackage{epsfig}
\usepackage{setspace}
\usepackage{color}
\renewcommand\today{\ifcase\month\or
  January\or February\or March\or April\or May\or June\or
  July\or August\or September\or October\or November\or December\fi
	\space\number\day, \number\year}
\newcount\minute
\newcount\hour
\def\currenttime{%
	\minute\time
	\hour\minute
	\divide\hour60
	\the\hour:\multiply\hour60\advance\minute-\hour\the\minute}
\def\draftnote{printed on \today, time \currenttime;\quad file name: \jobname}%
\newtheorem{theorem}{Theorem}[section]

\newtheorem{teo}[theorem]{Theorem}
\newtheorem{lm}[theorem]{Lemma}

\newtheorem{cor}[theorem]{Corollary}

\newtheorem{pr}[theorem]{Proposition}

\newtheorem{ex}[theorem]{Example}

\begin{document}
\title[Public-key cryptosystem based on invariants]{Public-key cryptosystem based on invariants of diagonalizable groups}
\thanks{This publication was made possible by a NPRF award NPRP 6 - 1059 - 1 - 208 from the Qatar National Research Fund (a member of The Qatar Foundation). The statements made herein are solely the responsibility of the authors.}
\author{Franti\v sek~ Marko}
\email{fxm13@psu.edu}
\address{Penn State Hazleton, 76 University Drive, Hazleton, PA 18202, USA}
\author{Alexandr N. Zubkov}
\email{a.zubkov@yahoo.com}
\address{Sobolev Institute of Mathematics, Siberian Branch of Russian Academy of Science (SORAN), Omsk, Pevtzova 13, 644043, Russia}
\author{Martin Jur\'a\v s}
\email{martinjuras@gmail.com}
\address{Savannah College of Art and Design, Department of General Education, Savannah, GA 31402, USA}
\begin{abstract}
We develop a public key cryptosystem based on invariants of diagonalizable groups and investigate properties of such cryptosystem first over finite fields, then over number fields and finally over finite rings.
We consider the security of these cryptosystem and show that it is necessary to restrict the set of parameters of the system to prevent various attacks (including linear algebra attacks and attacks based on Euclidean algorithm).
\end{abstract}
\maketitle
\section*{Introduction}

A new idea for a public-key cryptosystem based on the invariant theory was proposed by Grigoriev in \cite{dima1}. His original idea was later developed in the paper \cite{dima2}.
The last paragraph of the paper \cite{dima1} reads as follows:

"The current state of the art in cryptography does not allow one to prove the security of cryptosystems; this is usually a question of belief in the difficulty of a revelant problem and a matter of experience (that is why it is not quite unusual to have a paper on cryptography without theorems, for example, this paper).
Quite the opposite, one can expect a "disappointing" breaking of a particular cryptosystem. This can happen for any of the afforementioned examples (without solving the graph isomorphism problem, see the discussion above). On the other hand, such breaking could lead to interesting algorithms in the theory of group representations. Thus one can treat the above examples (and the general construction as a whole) just as a suggestion to play with cryptosystems based on the invariant theory."

The purpose of our paper is to develop and design a public-key cryptosystem based on invariants of diagonalizable groups.
We go beyond the philosophy of the preceeding quote and design a concrete public-key cryptosystem, present an algorithm for its implementation and
show how to break systems based on invariants of some groups.

At first we consider these cryptosystems over finite fields $F$, then we investigate cryptosystems over fields of characteristic zero (in particular, over number fields), and finally we work with cryptosystems over finite rings (in particular, residue classes of number fields modulo an integer).  Each part is distinguished by distinctive properties. For example, cyclicity of the multiplicative group $F^\times$ plays the most important role over finite fields, the theory of divisors and factorization properties are most important for the number fields, and both properties are important when we work over finite rings. One property that remains valid in all cases is that if $G$ not cyclic then it produces more complicated (and secure) cryptosystems as compared to the case when $C$ is cyclic.

Finding an invariant of the group $G$ is trivial in the finite field case, and what is challenging is to find one separating vectors from the given set $S$. We show a simple example when $n=2$ for which the security of the cryptosystem is equivalent to the discrete logarithm assumption. 
Over number fields, the main problem is to find an invariant of $G$ and the problem of the separation of elements of $S$ can be neglected.
The cryptosystems over finite rings combine features of the previous two cases and further investigation of their properties will be necessary. 

Finally, our work on this cryptosystem leads to an investigation of interesting mathematical problems related to the security of the invariant-based cryptosystem.
Theoretical results about related mathematical concepts like minimal degrees of invariants and invariants of supergroups will appear in a separate paper \cite{jmz2}.

\section{Invariants of finitely-generated linear groups}

In this paper, we will consider only finitely generated groups $G$ acting faithfully on a finite-dimensional vector space $V=F^n$ over a field $F$ of arbitrary characteristic.
Therefore we can asume that $G\subset GL(V)$. From the very beginning, assume that the representation $\rho:G\to GL(V)$ is fixed, and the group $G$ is given by a finite set of generators. With respect to the standard basis of $V$, each element $g$ of $G$ is therefore represented by an invertible matrix of size $n\times n$, and $g$ acts on vectors in $V$ by matrix multiplication.

Let $F[V]=F[x_1, \ldots, x_n]$ be the algebra of polynomial functions on $GL(V)$. Then $G$ acts on $F[V]$ via
$gf(v)=f(g^{-1}v)$, where $g\in G$, $f\in F[V]$ and $v\in V$.
An invariant $f$ of $G$ is a polynomial $f\in F[V]$
which has a property that its values are the same on orbits of the group $G$. In other words, for every vector $v\in V$ and for every element $g\in G$, we have $f(gv)=f(v)$.
We note that different representations of $G$ lead to different invariants in general but this is not going to be a problem for us since our representation of $G$ is fixed.
We will denote the algebra of invariants of $G$ by $F[V]^G$.

\section{Public key-cryptosystem based on invariants}\label{system}

We start by recalling the original idea of the public-key cryptosystem based on invariants from the paper \cite{dima1} and recalling its modification presented in \cite{dima2}.

\subsection{Design of cryptosystems based on invariants of groups}

To design a cryptosystem, Alice needs to choose a finitely generated subgroup $G$ of $GL(V)$ for some vector space $V=F^n$ and a set $\{g_1, \ldots, g_s\}$ of generators of $G$. 
Alice also chooses an $n\times n$ matrix $a$. Alice needs to know an invariant $f$ of this representation of $G$. Depending on this invariant $f$, Alice chooses a set $M=\{v_0, \ldots, v_{r-1}\}$ of vectors from $V$ such that the set $S=aM=\{av_0,\ldots, av_{r-1}\}$ is separated by the invariant $f$. This means that $f(av_i)\neq f(av_j)$ whenever $i\neq j$.
The set $M$ represents messages Alice can receive and elements of the set $S$ are bijectively assigned to elements of $M$. The sets $S$ is a part of the public key. 

Alice also chooses a set of randomly generated elements $g_1, \ldots, g_m$ of $G$ (say, by multiplying some of the given generators of $G$), which generates a subgroup of $G$ that will be denoted by $G_s$. 

Alice announces as a public key the set $S$ representing possible messages, and the group $H=a^{-1}G_sa$, conjugated to $G_s$, by announcing its generators $h_i=a^{-1}g_ia$ for $i=1, \ldots, m$.

In the first paper \cite{dima1} its author assumes that the group $G$, its representation in $GL(V)$ and the invariant $f$ are in the public key. We refer to this setup as {\it variant one}. However, the version in paper \cite{dima2} assumes that $G$, its representation in $GL(V)$ and the invariant $f$ are secret.  We refer to this setup as {\it variant two}. We will comment on both variants later.

For the encryption, every time Bob wants to transmit a message $m\in M$, he looks up the correponding element of $v_i\in S$ and chooses a randomly generated element $h$ of the group $H$(by multiplying some of the generators of $H$ given as a public key). 
Then he computes $u=hv_i$ and transmits the vector $u\in V$ to Alice.

To decript the message, Alice first computes $au$ and then applies the invariant $f$. If $u=hv_i$, then $f(au)=f(ahv_i)=f(aa^{-1}gav_i)=f(gav_i)=f(av_i)$. Since $a$ was chosen so that
$f(av_i)\neq f(av_j)$ whenever $i\neq j$, Alice can determine from the value of $f(au)$ whether the symbol $v_i$ and the corresponding message that was encrypted by Bob.

Let us discuss briefly the choices of $n$, $F$, $G$ and $S$.

It appears that choosing large $n$ would be better for the security of the cryptosystem but it would increase the expansion in size from plaintext to ciphertext.

The bigger and more complicated the structure of $F$, the better for the security of the cryptosystem. Analogously, the more complicated structure of $G$, in particular $G$ not cyclic would be preferred. 

Finally, we should choose the set $S$ as large as possible for two reasons. The first reason is that larger set $S$ shrink the number of invariants of $C$ that separate elements of $S$ and thus increases the security of the cryptosystem.
The second reason is that larger set $S$ decreases the ratio of the expansion in size from plaintext to ciphertext for the encryption using this cryptosystem.

\subsection{Previously described attacks on the cryptosystem based on invariants of groups}

Let us note that it is important that during the encryption process by Bob he uses all generators $h_i$ for scrambling the message. If some generators are not involved, then to decode his message Charlie would succeed if he finds an invariant of a subgroup of $H$, which is an easier task.

The attacks described below are mentioned in \cite{dima1} and \cite{dima2}. We are providing their description for the convenience of the reader and for further clarification. 
Also, these attacks we previously described only for the case when $|M|=2$ and we adapt them to the case when $|M|=r$.

To break the encryption, it is enough for Charlie to find any invariant $f'$ of the group $H$ that separates elements of $S$. 
If $r>2$ then we can replace this by a weaker condition. Namely, it is enough to find $f'$ such that $f'(u)=f'(v_i)$ for a unique $v_i\in S$.

Indeed if Charlie computes $f'(u)=f'(hv_i)=f'(v_i)$ and then compares $f'(u)$ with $f'(v_i)$. If there is a unique vector $v_i$ such that $f'(u)=f(v_i)$, then the message corresponding to $v_i$ was sent by Bob.

\subsubsection{Variant one}

Consider variant one of the cryptosystem - that is, the group $G$, its representation in $GL(V)$ and an invariant $f$ are known. We can assume that $f$ is a homogeneous polynomial of degree $d$. In this case, it is known that there is a homogeneous invariant $f'$ of $H$ of degree $d$ that is of the form $f'(v)=f(bv)$ for some matrix $b\in GL(V)$.
Then $f'$ is an invariant of $H$ if and only if $f(bh_iv)=f'(h_iv)=f'(v)=f(bv)$ for each generator $h_i$, where $i=1, \ldots, m$ of $H$.
Comparing coefficients at $\binom{n+d-1}{d}$ monomials we obtain $m\binom{n+d-1}{d}$ linear equations in $n^2$ variables (entries of $b$). Any solution of this system produces an invariant of $H$.

Another possible way to attack the system is to find a matrix $b\in GL(V)$ such that $bHb^{-1}\subset G$. This technique is related to the conjugacy problem for matrix groups and the graph isomorphism problem.

\subsubsection{Variant two}

In variant two of the cryptosystem, the group $G$, its representation in $GL(V)$ and the invariant $f$ are secret. However, Charlie can attempt to find an invariant $f'$ directly
by choosing a possible degree $d$ and solving linear systems derived from the equations $f'(h_iv)=f'(v)$ for each generator $h_i$, where $i=1. \ldots, m$.
This produces a linear system consisting of $m\binom{n+d-1}{d}$ equations in the $\binom{n+d-1}{d}$ unknowns that are the coefficients at monomials in $f'$.

Another approach is to find a matrix $h\in H$ such that $hu=v_i$ for some $i$ (attempting to recover the encryption done by Bob). This problem is related to the vector transporter problem and the graph isomorphism problem - see \cite{dima1}. Let us note that it was announced reently in \cite{bab} that the graph isomorphism problem 
can be solved in quasipolynomial time.

\section{Cryptosystems over finite fields $F$}

In this section we will discuss cryptosystems based in invariants of groups over finite fields $F$. We will present concrete examples and show how the security of those cryptosystems 
is guaranteed if we assume computational hardness of the dicrete logarithm problem.

\subsection{$n=1$}

For simplicity, in the case $n=1$, we will assume that the cardinality of the set $S$ is $2$.

The case $n=1$ is singular and it implies $G\subset F$. If there is a nonconstant invariant $f=p(x)$ of $G$ that attains the constant value $c$ when evaluated on each element of $G$, then $G$ is a subset of the set of roots of the polynomial $p(x)-c$. In particular, $G$ is finite.

Let $F$ be a finite field $GF(q)$ of characteristic $p>0$ and cardinality $q=p^b$. The set of non-zero elements $F^{\times}$ of $F$ with respect to the multiplication in $F$ form a cyclic group generated by a primitive element $a$. Also, $F$ is isomorphic to the splitting field of the polynomial $x^q-x=0$ over the prime field $GF(p)$.
In particular, $x^q-x$ is an invariant of $F$ and $x^{q-1}$ is an invariant of $F^{\times}$ attaining the value $1$.

Since $F^{\times}$ is cyclic, every subgroup $G_d$ of $F^{\times}$ is also cyclic, generated by $a^{\frac{q-1}{d}}$, given as a set of roots of $x^d-1=0$ for some $d|(q-1)$.  The set of invariants of $G$ is generated by $x^d$.
If $d\neq q-1$, we can choose $v_0$ and $v_1$ from $F^{\times}$ such that $v_0^d\neq v_1^d$. In particular, any choice $v_0\in G_d$ and $v_1\notin G_d$ will do. Then the invariant $f=x^d$ of $G$ separates $v_0$ and $v_1$.

For setting up the corresponding cryptosystem over a finite field $F$ we need to select an element $g\in F$ (that generates $G$) and find an exponent $d'<q-1$ such that $g^{d'}=1$.
Then we can select for $v_0$ any element of $G$ (that is a power of $g$). We also need to find an element $v_1\in F$ that does not belong to $G$. Since we do not know the order $d$ of
$g$, the simplest way to guarantee this condition is to make sure that $v_1^{d'}\neq 1$. Then $x^{d'}$ is an invariant of $G$ separating $v_0$ and $v_1$.

To break such a system, we need, for a given group $G$ and vectors $v_0$ and $v_1$, to find an invariant of $G$ separating $v_0$ and $v_1$. An invariant of $G$, namely
the polynomial $x^{q-1}$ is known from the beginning but the problem is to find one that separates $v_0$ and $v_1$.

If there is a algorithm that determines the order any element $g\in F$ in polynomial time, then the cryptosystem can be setup and broken in polynomial time, hence it is not secure.
Even if there is no algorithm that determines the order of $g$ in polynomial time, it might be possible to find a separating invariant and break the cryptosystem randomly.

Consider the following example.

\begin{ex}
Let $s$ be a prime dividing $q-1$, $d=\frac{q-1}{s}$, $G=G_d$ is generated by $g$, $v_0=g$ and $v_1=a$ ($a$ is a primitive element of $F^{\times}$).
Then the only invariant of $G$ that separates $v_0$ and $v_1$ is $x^d$. Therefore breaking of this cryptosystem is equivalent to finding of the order of $g$, and is also equivalent to finding of the prime factor $s$ of $q-1$. Therefore we conclude that breaking of all cryptosystems of this type is equivalent to finding of all prime divisors of $q-1$ (hence finding of the prime factorization of $q-1$.)
\end{ex}

We give a brief review of the computational complexity of factorization of integers in the next subsection. 

The value of the above example is in showing that even if an invariant of $G$ is known, it might not be completely trivial to find an invariant of $G$ separating $v_0$ and $v_1$.

Of course we can find the order of $g$ and break this crytosystem using the discrete logarithm (although it might be easier just to find the prime factorization of $q-1$).
Since the multiplicative group of $F$ has a primitive element $a$, and the multiplicative group of $G$ has generator $g$,
we can use the discrete logarithm to find out the exponent $h$ such that $g=a^h$. Then $g$ is the primitive $d=\frac{q-1}{GCD(h, q-1)}$-th root of unity.
Once we know the order $d$ of $g$, we have found the invariant $x^d$, separating $v_0$ and $v_1$.

Nevertheless, we will not use the case $n=1$ to setup a cryptosystem due to concerns about its security.
The reason is that in the original setup of the system some partial information about the order $d$ of $g$ is required. Only a partial information about the order $d$ is required to break such a cryptosystem and we are not aware of an effective setup when it is easy to create such a cryptosystem and difficult to break it.


\subsection{Computational complexity of the factorization of integers}

We refer to \cite{vas} for the description of various algorithms and their complexity.

First we will overview the determininistic algorithms for factorization of integers.

One of the simplest is the Fermat algorithm that works fast $n=pq$ is a product of two primes that are of the same magnitude.
The most popular detemininistic algorithms for factorization of integers (all of them of exponential complexity) are $(p-1)$-method, Pollard's $\rho$-method and the Pollard-Strassen algorithm. These algorithms are often used to find small prime factors. For more details and description of other algorithms, see Chapter 2 of \cite{vas}.

When working over finite field $F=GF(q)$, where $q=p^k$ for a prime $p$, the following theorem helps to determine the factorization of $q-1$.

\begin{teo}(Theorem 2.21 of \cite{vas})
Let $b,k\in \mathbb{N}$, $b>1$, and $n=b^k-1$. If $p$ is a prime number dividing $n$, then one of the following two assertions holds:

1. $p|b^d-1$ for some $d<k, d|k$;

2. $p\equiv 1 \pmod k$.

If $p>2$ and $k$ is odd, then, in the second case $p\equiv 1 \pmod{2k}$.
\end{teo}

Although this statement seems easy to use, in reality it gives an algorithm of exponential complexity.

Probabilistic algorithms for factorization of integers with subexponential complexity are discussed in details in Chapter 3 of \cite{vas}. The complexity of these algorithms is of the form $L_n[\gamma, c]$, where $\gamma=\frac12$ or $\frac13$ and $c$ is a positive constant, where
\[L_x[\gamma,c]=e^{(c+o(1))(log \, x)^{\gamma}(log \, log \, x)^{1-\gamma}}\]
and $o(1)\to 0$ as $n\to \infty$.

Most popular algorithms of this nature are Lenstra eliptic curve method, quadratic sieve and number field sieve. For more details see Chapters 3 and 4 of \cite{vas}.

\subsection{$n=2$ and $G$ cyclic}

Next, we will discuss the case when $G$ is cyclic and $n=2$ and show that breaking of the corresponding cryptosystem is equivalent to solving of a discrete logarithm problem. 

Assume that $F=FG(q)$ is a finite field of cardinality $q$, where $q=p^r$ and $a$ be a primitive element of $F^\times$.

Let a cyclic group $G$ be generated by the element 
$\gamma=\begin{pmatrix} \gamma_1 &0\\0&\gamma_2\end{pmatrix}$, where $\gamma_1=a^{l_1}$ and $\gamma_2=a^{l_2}$ and 
the set $S$ consists of vectors $\vec{v}_i=\begin{pmatrix}a_{i1}\\a_{i2}\end{pmatrix}$ for $i=1, \ldots, s$.

Since every invariant of $G$ is a sum of monomial invariants, to obtain a complete description of all invariants of $G$ we only need to find monomial invariants.
A monomial $f=x_1^{d_1}x_2^{d_2}$ is an invariant of $G$ if $\gamma_1^{d_1x}\gamma_2^{d_2x}$ is constant for every integer $x$. Plugging in $x=q-1$ we get
$\gamma_1^{d_1}\gamma_2^{d_2}=1$ which is equivalent to 
\[d_1l_1+d_2l_2\equiv 0 \pmod{q-1}.\]

A monomial invariant $f=x_1^{d_1}x_2^{d_2}$ separates vectors $\vec{v}_i$ and $\vec{v}_j$ if and only if 
\[a_{i1}^{d_1}a_{i2}^{d_2}\neq a_{j1}^{d_1}a_{j2}^{d_2}.\]

The following is an example of a group $G$ for which finding an invariant $f$ separating vectors from the given set $S$ 
is computationally hard problem - see Proposition \ref{equivalence}.

\begin{ex}\label{ex3.3}
Assume that $F=FG(q)$ is a finite field of cardinality $q$, where $q=p^r$ and $s$ is a (large) prime that divides $q-1$.
Let $\alpha\in F^{\times}$ be an element of order $s$, 
and let $\beta=\alpha^b$ for a secret integer $b$ not divisible by $s$.

Let $V=F^2$, $G\subset GL(V)$ be a cyclic group generated by the element
$g=\begin{pmatrix} \alpha & 0\\ 0&\beta \end{pmatrix}$, and the set $S$ consists of vectors 
$v_i=\begin{pmatrix} 1 \\ a_i\end{pmatrix}$, where $a_i=\alpha^i$ for $i=0, \ldots, s-1$.
Consider the cryptosystem based on this group $G$ and the set $S$.
\end{ex}

A general element of $G$, written as $g^x$ for some exponent $x$, acts on vectors $v_i$ as
$g^x v_i=\begin{pmatrix} \alpha^x\\ \beta^x a_i \end{pmatrix}=\begin{pmatrix} w_1\\w_2 \end{pmatrix}=\vec{w}$, for $i=0, \ldots, s-1$.

\subsubsection{Discrete logarithm approach}
If the exponent $b$ is known, then we can decode $a_i$ and $v_i$ from $\vec{w}$ simply using $a_i=\frac{w_2}{w_1^b}$.

Otherwise, we can break this cryptosystem if we have an effective algorithm for the discrete logarithm problem in $F^\times$.
Indeed, to determine $i$ it is enough to solve for $y$ such that $(\frac{\beta}{\alpha})^y = \frac{w_2}{w_1a_i}$.

\subsubsection{Invariant approach}

\begin{pr}\label{equivalence}
Finding a monomial invariant $f$ of $G$ from Example \ref{ex3.3} separating elements of $S$ is equivalent to finding $b \pmod s$, hence to finding a solution of the discrete logarithm problem for the pair $\alpha$ and $\beta$. 
\end{pr}
\begin{proof}
We have chosen vectors $v_i$ in such way that none of the invariants of $G$ of the form
$x_1^e$ or $x_2^e$ would separate any two of them.
Since there is a monomial invariant of $G$ separating all elements of the set $S$ (for example $x_1^bx_2^q$), we can assume that it is of the form $f=x_1^{d_1}x_2^{d_2}$, where
$d_1,d_2\neq 0$.
Then
\[d_1+bd_2 \equiv 0 \pmod s \text{ and } a_i^{d_2}\neq a_j^{d_2} \text{ whenever } i\neq j.\]
 This implies
$f(\vec{w})=w_1^{d_1}w_2^{d_2}=(\alpha^x)^{d_1}(\beta^xa_i)^{d_2}=(\alpha^x)^{d_1+ad_2}a_i^{d_2}=a_i^{d_2}$ and that determines $v_i$.

Therefore finding a monomial invariant $f$ separating elements of $S$ requires finding a solution of the congruence $d_1+bd_2 \equiv 0 \pmod s$ which is equivalent to determining $b \pmod s$ and this is equivalent to the discrete logarithm problem for the pair $\alpha$ and $\beta$. 

Conversely, if we are able to determine $b \pmod s$, then $f_r=x_1^bx_2^{-1}$ is a rational invariant of $G$ and $x_1^bx_2^q$ is a monomial invariant of $G$ separating all 
elements of $S$.
\end{proof}

Let us describe an algorithm how to describe all invariants of cyclic group $G$ when $n=2$
and find one separating elements of $S$.
Finding a monomial invariant $x_1^{d_1}x_2^{d_2}$ of $G$ is equivalent to solving the congruence 
$d_1l_1+d_2l_2\equiv 0 \pmod{q-1}$. To find such an invariant, we first need to find the primitive element $a$ and use discrete logarithms to solve for $l_1$ and $l_2$
from $\gamma_1=a^{l_1}$ and $\gamma_2=a^{l_2}$.

Once $l_1$ and $l_2$ are known, there is an effective way to describe all solutions of the above congruence and all monomial invariants of $G$ as follows.
If $GCD(l_1, q-1)=1$, then we can choose any $d_2$ and compute $d_1=-l_1^{-1}d_2l_2 \pmod{q-1}$. Here we find $l_1^{-1}\pmod{q-1}$ using the Euclidean algorithm which runs in polynomial time.
If $GCD(l_1, q-1)=d>1$, then $d|d_2l_2$,  We replace $l_1$ by $\frac{l_1}{d}$, $l_2$ by $l_2'=\frac{l_2}{GCD(l_2,d)}$, $d_2$ by $d_2'=d_2\frac{GCD(l_2,d)}{d}$ and 
the congruence $d_1l_1+d_2l_2\equiv 0 \pmod{q-1}$ by
$d_1l_1'+d_2'l_2'\equiv 0 \pmod{\frac{q-1}{d}}$. The solutions of the last congruence are described analogously as above since $GCD(l_1', \frac{q-1}{d})=1$. Finally, we compute
$d_2=d_2'\frac{d}{GCD(l_2,d)}$. Since we know all invariants, it remains to select one separating elements of $S$. 

\subsubsection{Complexity and expansion in size}

Coming back to the Proposition \ref{equivalence}, we have seen that even in the simplest case when $n=2$ and the group $G$ is cyclic the task of finding invariants of $G$ separating elements of $S$ is of the same complexity as the discrete logarithm problem.

Since the discrete logarithm assumption is weaker than the computational Diffie-Hellman assumption and that is weaker than decisional Diffie-Hellman assumption which are used in 
Diffie-Hellman key exchange and ElGamal encryption, we have a guarantee that the invariant-based cryptosystem of Example \ref{ex3.3} is at least as secure as many standard and widely used public-key cryptosystems.

If we compare the cryptosystem from Example \ref{ex3.3} to ElGamal public-key cryptosystem, we note that while ElGamal encryption produces a 2:1 expansion in size from plaintext (by this we mean a sequence of 0's and 1's) to ciphertext, and the invariant-based cryptosystem from the above example produces $2\log_2 q : \log_2 s$ expansion in size from plaintext to ciphertext. If $s$ is not small in comparison to $q$, say $\frac{\log_2 s}{\log_2 q}$ is close to 1, then this expansion will be close to 2:1.
Also, higher the number $s$ becomes, the ratio of expansion from the plaintext to the ciphertext would improve.  

In general, if $s$ is small or if it is comparable to $q$, then finding the prime $s$ would be easy.  
If $q-1=s_1s_2s_3$, where $s=s_1\neq s_2$ are two large primes and $s_3$ is small integer, then it could be challenging to determine the order $s$ of the element $\alpha$ since one needs to determine the prime factorization of $q-1$. In this cases the best choice would be if $\frac{\log_2 s}{\log_2 q}$ is close to $\frac12$. This implies that the corresponding expansion from plaintext to ciphertext will be close to 4:1 (Perhaps not a big price to pay for increasing the security of the cryptosystem).

In the case of general $n$ and a set $S$ of cardinality $r$ we get $n\log_2 |F|:\log_2 r$ expansion. 
A trivial upper bound on the cardinality $r$ of the set $S$ is given by the index $[F^n:G]$. Adjusting the previous example for $n=2$ we are able to get the expansion ratio close to 
$n:1$ - and quite possibly even better ratio with different choices of $S$. While such ratio is a disadvantage, if we choose $n$ small, this should not play a big role for the effectiveness of the cryptosystem.

\subsection{Computational complexity of the discrete logarithm problem}

Algorithms for computing discrete logarithm are subject of Chapter 5 of \cite{vas}. There are deterministic algorithms of exponential complexity and probabilistic algorithms of subexponential complexity. There is an algorithms for discrete logarithm problems in prime fields of complexity $L_p[\frac12;c]$.
ElGamal \cite{elgam} gave an algorithm that works over Galois fields $GF(q)$ and has the complexity $L_q[\frac12;c]$. It is interesting that it uses a representation of $GF(q)$ as the residue class ring $Z/\mathfrak{P}$, where $Z$ is the ring of algebraic integers of the cyclotomic field $\mathbb{Q}(\zeta_p)$, and $\mathfrak{P}$ is a prime ideal of $Z$ of norm $q$.

Another algorithm working in prime fields based on number field sieve has complexity $L_p[\frac13;c]$. Finally, in section 5.6 of \cite{vas} there is an interesting algorithm for discrete logarithm with composite modulus based on Fermat quotients that works in residue class rings $\mathbb{Z}/m\mathbb{Z}$ with composite $m$.

\subsection{$G$ cyclic}

When $n>2$, the problem of finding monomial invariants of a cyclic group $G$ over a finite field $F$ is essentially reduced to multiple applications of the method already used to find invariants of cyclic $G$ in the case $n=2$. 

Assume that $F=FG(q)$ is a finite field of cardinality $q$, where $q=p^r$ and $a$ be a primitive element of $F^\times$.

Let a cyclic group $G$ be generated by the element 
$g=\begin{pmatrix} \gamma_1 &0 & \ldots &0 \\0&\gamma_2& \ldots&0\\\ldots &\ldots&\ldots&\ldots\\0&0&\ldots&\gamma_n\end{pmatrix}$, 
where $\gamma_j=a^{l_j}$ and $0\leq l_j<q-1$ for $j=1, \ldots, n$ and the set $S$ consists of vectors $\vec{v}_i=\begin{pmatrix}a_{i1}\\a_{i2}\\\ldots \\a_{in}\end{pmatrix}$ for $i=1, \ldots, s$.

Since every invariant of $G$ is a sum of monomial invariants, to obtain a complete description of all invariants of $G$ we only need to find monomial invariants.
A monomial $f=x_1^{d_1}x_2^{d_2}\ldots x_n^{d_n}$ is an invariant of $G$ if $\gamma_1^{d_1x}\gamma_2^{d_2x}\ldots \gamma_n^{d_nx}$ is constant for every integer $x$. 
Plugging in $x=q-1$ we get
$\gamma_1^{d_1}\gamma_2^{d_2}\ldots \gamma_n^{d_n}=1$ which is equivalent to 
\[d_1l_1+d_2l_2+\ldots + d_nl_n\equiv 0 \pmod{q-1}.\]
In particular, $x_i^{l_j}x_j^{-l_i}$ for $i\neq j$ are rational invariants of $G$ and $x_i^{l_j}x_j^{q-1-l_i}$ are monomial invariants of $G$. 

A monomial invariant $f=x_1^{d_1}x_2^{d_2}\ldots x_n^{d_n}$ separates vectors $\vec{v}_i$ and $\vec{v}_j$ if and only if 
\[a_{i1}^{d_1}a_{i2}^{d_2}\ldots a_{in}^{d_n}\neq a_{j1}^{d_1}a_{j2}^{d_2}\ldots a_{jn}^{d_n}.\]

\begin{pr}\label{reduce}
Let $G$ be a cyclic group, $\vec{v}_k\neq\vec{v}_l$ be elements of $F^n$ and $n\geq 2$. If there is a monomial invariant $f$ of $G$ separating $v_k$ and $v_l$, then there is a 
monomial invariant of $G$ the form $x_i^{e_i}x_j^{e_j}$ that also separates $\vec{v}_k$ and $\vec{v}_l$.
\end{pr}
\begin{proof}
We will proceed by induction on $n$. The statement is trivial for $n=2$. Assume that the statement is true for all $k<n$. If we write $f=x_1^{d_1}\ldots x_n^{d_n}$, then 
$f$ invariant separating $\vec{v}_k$ and $\vec{v}_l$ implies 
\[d_1l_1+\ldots +d_nl_n \equiv 0 \pmod{q-1} \text{ and } a_{k1}^{d_1}\ldots a_{kn}^{d_n}\neq a_{l1}^{d_1}\ldots a_{ln}^{d_n}.\]   
Denote $d=GCD(l_1, \ldots, l_n, q-1)$. Then there is an index $t$ such that $(a_{k1}^{d_1}\ldots a_{kn}^{d_n})^{\frac{l_t}{d}}\neq (a_{l1}^{d_1}\ldots a_{ln}^{d_n})^{\frac{l_t}{d}}$.
If the invariant $x_u^{-\frac{l_t}{d}}x_t^{\frac{l_u}{d}}$ does not separate $\vec{v}_k$ and $\vec{v}_l$, then 
$a_{ku}^{-\frac{l_t}{d}}a_{kt}^{\frac{l_u}{d}}=a_{lu}^{-\frac{l_t}{d}}a_{lt}^{\frac{l_u}{d}}$ which gives
$a_{ku}^{-d_u\frac{l_t}{d}}a_{kt}^{d_u\frac{l_u}{d}}=a_{lu}^{-d_u\frac{l_t}{d}}a_{lt}^{d_u\frac{l_u}{d}}$. This together with  
$a_{k1}^{d_1\frac{l_t}{d}}\ldots a_{kn}^{d_n\frac{l_t}{d}}\neq a_{l1}^{d_1\frac{l_t}{d}}\ldots a_{ln}^{d_n\frac{l_t}{d}}$ implies
\[a_{k1}^{d_1\frac{l_t}{d}}\ldots a_{ku}^0 \ldots a_{kt}^{d_t\frac{l_t}{d}+d_u\frac{l_u}{d}} \ldots a_{kn}^{d_n\frac{l_t}{d}}
\neq a_{l1}^{d_1\frac{l_t}{d}}\ldots a_{lu}^0 \ldots a_{lt}^{d_t\frac{l_t}{d}+d_u\frac{l_u}{d}} \ldots a_{ln}^{d_n\frac{l_t}{d}}.\] 

Replace the space $F^n$ by $F^{n-1}$, the group $G$ with a group $G'$ generated by the matrix $g'$ obtained from the matrix $g$ generating $G$ by deleting its $u$-th row and column, and replace vectors $\vec{v}_i$ with $\vec{v}_i'$ obtained by deleting the entry in their $u$-th row. 
By the above, the monomial $f'=x_1^{d_1\frac{l_t}{d}}\ldots \widehat{x_u^0} \ldots x_t^{d_t\frac{l_t}{d}+d_u\frac{l_u}{d}}\ldots x_n^{d_n\frac{l_t}{d}}$ 
separates $\vec{v}_k'$ and $\vec{v}_l'$. Additionally, $f'$ is an invariant of $G'$ because
\[\begin{aligned} &d_1\frac{l_t}{d}l_1+\ldots +\widehat{0l_u}+ \ldots +(d_t\frac{l_t}{d}+d_u\frac{l_u}{d})l_t+\ldots +d_n\frac{l_t}{d}l_n =\\
&\frac{l_t}{d}\Big[d_1l_1+\ldots + \widehat{0} + \ldots +(d_tl_t+d_ul_u)+ \ldots +d_nl_n \Big]\equiv 0 \pmod{q-1}.\end{aligned}\]

Using the inductive assumption we get an invariant of $G'$ of the form $x_i^{e_i}x_j^{e_j}$ that separates $\vec{v}_k'$ and $\vec{v}_l'$. It is clear that this is also 
an invariant of $G$ and that it separates $\vec{v}_k$ and $\vec{v}_l$.
\end{proof}

As a consequence of the above proposition we conclude that in order to find an invariant monomial of $G$ separating vectors $\vec{v}_k$ and $\vec{v}_l$ of $S$ we can proceed by going through all the pairs of indices $i$ and $j$ from $1$ through $n$ and and solve analogous problem when $G$ is replaces by $G_{ij}$ generated by $2\times 2$ matrix 
$g_{ij}=\begin{pmatrix}\gamma_i&0\\0&\gamma_j\end{pmatrix}$ and vectors $\vec{v}_k$ and $\vec{v}_l$ are replaced by vector $\begin{pmatrix}a_{ki}\\a_{kj}\end{pmatrix}$ and 
$\begin{pmatrix}a_{li}\\a_{lj}\end{pmatrix}$, respectively.

\subsection{$G$ not cyclic}

Noncyclic groups $G$ exist for every $n\geq 2$. We start with an example generalizing Example \ref{ex3.3}.

\begin{ex}\label{ex3.6}
Let $F=FG(q)$, where $q=p^r$, and $s_1$ and $s_2$ be (large) primes dividing $q-1$. Let $\alpha_1\in F^\times$ be an element of order $s_1$ and $\alpha_2\in F^\times$ 
be an element of order $s_2$, and $a_i$ be distinct elements of $F^{\times}$ such that their order divides $s_1s_2$.
Let $G$ be given by two generators $g_1=\begin{pmatrix} \alpha_1 & 0\\ 0&\beta_1 \end{pmatrix}$ and $g_2=\begin{pmatrix} \alpha_2 & 0\\ 0&\beta_2 \end{pmatrix}$,
where $\beta_1=\alpha_1^{b_1}$ and $\beta_2=\alpha_2^{b_2}$ for secret integers $b_1$ not divisible by $s_1$, and $b_2$ not divisible by $s_2$
Let the set $S$ consist of vectors $\vec{v}_i=\begin{pmatrix} 1 \\ a_i\end{pmatrix}$.
Consider the cryptosystem based on this group $G$ and set $S$.
\end{ex}

The general element $g$ of $G$, written as $g_1^{x_1}g_2^{x_2}$, acts on $\vec{v}_i$ as
$g \vec{v}_i=\begin{pmatrix} \alpha_1^{x_1}\alpha_2^{x_2}\\ \beta_1^{x_1}\beta_2^{x_2} a_i \end{pmatrix}=\begin{pmatrix} w_1\\w_2 \end{pmatrix}=\vec{w}$.

As before, we have chosen vectors $\vec{v}_i$ in such way that none of the invariants of $G$ of the form $x_1^e$ or $x_2^e$ would separate them.
If there is an invariant of $G$ separating vectors $\vec{v}_i$, then we can assume that it is of the form $f=x_1^{d_1}x_2^{d_2}$.

If $f=x_1^{d_1}x_2^{d_2}$ is an invariant of $G$, then $(\alpha_1^{x_1}\alpha_2^{x_2})^{d_1}(\beta_1^{x_1}\beta_2^{x_2})^{d_2}=1$. 
If $a$ is the primitive element of $F^\times$, $\alpha_1=a^{a_1}$ and $\alpha_2=a^{a_2}$ for secret integers $a_1$ and $a_2$, then this is equivalent to
\[(a_1x_1+a_2x_2)d_1+(a_1b_1x_1+a_2b_2x_2)d_2 \equiv 0 \pmod {q-1}\]
for every $x_1$ and $x_2$.
This condition is equivalent to the system of congruences
\[a_1(d_1+b_1d_2) \equiv 0 \pmod {q-1} \text{ and  } a_2(d_1+b_2d_2)\equiv 0 \pmod {q-1}.\]
These two congruences are related to different generators $g_1$ and $g_2$ of $G$.
We have seen earlier that using discrete logarithms we can describe all monomial invariants of the cyclic subgroups $\langle g_1 \rangle$ and $\langle g_2 \rangle$ of $G$.
Of course, the invariants of $G$ are exactly those polynomials that are invariants with respect to $\langle g_1 \rangle$ and $\langle g_2 \rangle$ simultaneously.

The last two congruences can be solved by using integer linear programming because they are equivalent to the linear system
\[\begin{aligned} a_1&d_1 +&a_1b_1d_2+(q-1)&d_3 =&0\\
a_2&d_1+&a_2b_2d_2+(q-1)&d_4=&0
\end{aligned}\]
in integers $d_1$, $d_2$, $d_3$ and $d_4$.

Now consider the general case when $G$ has generators 
$g_i=\begin{pmatrix} \gamma_{i1} &0 & \ldots &0 \\0&\gamma_{i2}& \ldots&0\\\ldots &\ldots&\ldots&\ldots\\0&0&\ldots&\gamma_{in}\end{pmatrix}$ for $i=1, \ldots t$, 
where $\gamma_{ij}=a^{l_{ij}}$ and $0\leq l_{ij}<q-1$ for $j=1, \ldots, n$ and the set $S$ consists of vectors 
$\vec{v}_i=\begin{pmatrix}a_{i1}\\a_{i2}\\\ldots \\a_{in}\end{pmatrix}$ for $i=1, \ldots, s$.
The general element of $G$ is written as $g=g_1^{y_1}\ldots g_t^{y_t}$ for some integers $y_1, \ldots, y_t$.

Since every invariant of $G$ is a sum of monomial invariants, to obtain a complete description of all invariants of $G$ we only need to find monomial invariants.
A monomial $f=x_1^{d_1}x_2^{d_2}\ldots x_n^{d_n}$ is an invariant of $G$ if 
$\prod_{i=1}^t \gamma_{i1}^{y_id_1}\gamma_{i2}^{y_id_2}\ldots \gamma_{in}^{y_id_n}=1$ for all integers $y_1, \ldots, y_t$. This implies 
$\gamma_{i1}^{d_1}\gamma_{i2}^{d_2}\ldots \gamma_{in}^{d_n}=1$ for each $i=1, \ldots, t$ which is equivalent to the system of congruences
\[d_1l_{i1}+d_2l_{i2}+\ldots + d_nl_{in}\equiv 0 \pmod{q-1}\]
for each $i=1, \ldots, t$.

If the numbers $l_{ij}$ are determined (say using discrete logarithms), then the last system can be solved using integer programming since it is equivalent to 
the system 
\[d_1l_{i1}+d_2l_{i2}+\ldots + d_nl_{in}+d_{n+i}(q-1)=0\]
for each $i=1, \ldots, t$ in integer variables $d_1, \ldots, d_{n+t}$.

A monomial invariant $f=x_1^{d_1}x_2^{d_2}\ldots x_n^{d_n}$ separates vectors $\vec{v}_i$ and $\vec{v}_j$ if and only if 
\[a_{i1}^{d_1}a_{i2}^{d_2}\ldots a_{in}^{d_n}\neq a_{j1}^{d_1}a_{j2}^{d_2}\ldots a_{jn}^{d_n}.\]

We have seen that, for $F$ a finite field, the fact that $F^{\times}$ is cyclic allows the use of the discrete logarithm, which is computationally difficult but standard cryptographic tool.

If $G$ was cyclic it was enough to find any invariant of $G$ randomly and check if it separates $S$. If $G$ is not cyclic then more systematic knowledge of invariants of cyclic subgroups of $G$ is necessary and breaking of the cryptosystem based on noncyclic group $G$ seems more complicated than the case of the cyclic group $G$.
Additionally, it is not clear if there is a separating invariant based on two variables analogous as in Proposition \ref{reduce} in the case of noncyclic group $G$.
Therefore using noncyclic $G$ gives an advantage from the point of view of the security of the cryptosystem based on invariants of $G$.
 
The setup will be even more complicated if the underlying structure of $F$ is not cyclic. After we investigate the minimal degree of polynomial invariants (question related to a linear algebra attack on the cryptosystem) in the next section, we turn our attention to fields $F$ of characteristic zero. We will introduce a cryptosystem the security of which will depend on the fastorization properties in the number field $F$. Afterward, we will use residue classes of $F$ and replace $F$ by a finite commutative ring $R$ with a complicated multiplicative structure that will not allow an obvious use of the discrete logarithms. In the cases of the number field $F$ and the residue ring $R$, the nature of finding invariants of $G$ is different from discrete logarithm problem. This is clear in the case of a number field $F$. The residue residue ring $R$ have divisors of zeros and the multiplicative group  $U$ of its units is not cyclic. 

\subsubsection{Encryption based on discrete logarithm one-way functions}

We would like to make a small detour from the invariant-based cryptosytems and discuss cryptosystems based on discrete logarithms that are inspired by Examples 
\ref{ex3.3} and \ref{ex3.6}.

Assume that Alice chooses a finitely generated group $G$ acting on a set $M$. Let $\{g_1, \ldots, g_m\}$ be a set of generators of $G$. Alice chooses a subset $M_0\subset M$ such that any orbit $Gm$ for $m\in M$ intersects $M_0$ in  exactly one point. (The set of all blocks of plaintext that can be transmitted by Bob is injectively mapped to $M_0$.) Alice chooses a 
map $f:M\rightarrow M_0$ that is constant on each orbit $Gm$ of $G$ and retains it as a private key. Obviously $f$ restricted on $M_0$ is an identity.
She announces, as a public key, the (effectively described) set $M_0$ and the group $G$, by announcing its generators $g_1, \ldots, g_m$.

To encode a block of plaintext $m\in M_0$, Bob choses a random  element $g\in G$ (by multiplying some of the generators $g_1, \ldots, g_m$), and computes $m'=gm$ which he transmits to Alice.

Alice decrypts the message by applying the map $f$ as $f(m')=f(gm)=f(m)=m$.

\begin{ex} 
Consider the ElGamal cryptosystem with cyclic group $C$ of order $n$, generator $\alpha\in C$, private key $b\in \{0,1\ldots, n-1\}$ and the public key $\{\alpha, \beta, n\}$, where 
$\beta=\alpha^b$. The group $C$ coincides with the set of all blocs of plaintext that can be send by Bob to Alice. 
A cryptosystem is contructed as follows. 
Let $M=C\times C$, considered as a group with respect to the mutiplication induced by the diagonal action of $C$, and let $G$ to be its cyclic subgroup generated by $(\alpha,\beta)$.
Then $G$ acts on $M$ by multiplication.  We set $M_0=C$ and the map $f:M\rightarrow M_0$ to be $f(x,y)=yx^{-b}$.  
\end{ex}

\begin{ex} Let $A$ be an abelian group generated by $\alpha_1,\alpha_2\in A$.  Let $\beta=\alpha_1^{b_1}\alpha_2^{b_2}$ for some $b_1,b_2 \in \mathbb N$.
Let $M=A\times A\times A$ and $G$ be its cyclic subgroup generated by the element  $(\alpha_1,\alpha_2,\beta)$.  
We set $M_0=A$ and the map $f:M\rightarrow M_0$ to be $f(x,y,z)=zx^{-b_1}y^{-b_2}$. 
\end{ex}

Alice announces the group $A$ and the vector $(\alpha_1,\alpha_2,\beta)$ as a public key. 
To encode a block of plaintext $m\in A$, Bob chooses a random  number $e\in\{0,1,\ldots, n-1\}$ and transmits the vector 
$(\alpha_1,\alpha_2,\beta)^{e}(1,1,m)=(\alpha_1^{e},\alpha_2^{e},\beta^{e}m)$ to Alice.

Alice decrypts the message $m$ as
$f(\alpha_1^{e},\alpha_2^{e},\beta^{e}m)=\beta^{e}m\alpha_1^{-eb_1}\alpha_2^{-eb_2}=m$.

This encryption produces a 3:1 expansion in size from plaintext to ciphertext. 

The security of this cryptosystem depends on the ability of the evesdropper Charlie to solve the equation $\beta=\alpha_1^{x_1}\alpha_2^{x_2}$, for integers $x_1$ and $x_2$.
If we work over a finite field $F=FG(q)$, then  we can we can use the discrete logarithms in $F^{\times}$ to express $\alpha_1, \alpha_2$ and $\beta$ as powers of the primitive element $a$ of 
$F^{\times}$, say $\alpha_i=a^{e_i}$ and $\beta=a^e$. Then the equation $\beta=\alpha_1^{x_1}\alpha_2^{x_2}$ reduces to a congruence 
$e=e_1x_1+e_2x_2 \pmod{q-1}$.
Therefore we require the discrete logarithm assumption to guard against this attack.

\section{Minimal degree of polynomial invariants of $G$}\label{inv}

When we described the design of the cryptosystem based on invariants we have already remarked that its security depends on the difficulty of finding an invariant $f'$ of the group $H$ separating vectors in $S$. 

When we are working over a ground field $F$ of characteristic zero, then the condition that $f'$ separates $v_i$ and $v_j$ for $i\neq j$ 
might not be difficult to satisfy because the set of polynomials in $F[V]$, that take on different values when evaluated at $v_i$ and $v_j$, is open in the Zariski topology.
Therefore it is likely that a randomly chosen invariant $f'$ of $H$ will separate elements of $S$ in this case. 
Therefore when $F$ has characteristic zero, we need not be concerned whether $f'$ separates vectors from $S$.

This is contrary to the situation over finite fields when it was easy to find an invariant of $G$ but difficult to satisfy the condition that it separates elements of $S$.

\subsection{Guarding against the linear algebra attack}\label{LA}

Denote by $M_{G,V}$, or simply by $M_G$ or $M$ if we need not emphasise the group $G$ or the vector space $V$ it is acting on
the minimal positive degree of an invariant from $F[V]^G$. That is
$M_{G, V}=\min\{d>0 | F[V]_d^G\neq 0\}$. If $F[V]^G=F$, then we set $M_{G,V}=\infty$.

The notion of the minimal positive degree of an invariant and the value of $M=M_{G,V}$ are important for the security of the invariant-based cryptosystem (both variants one and two) we are considering.
For example, if we know that $M_G$ is so small that $m\binom{n+M-1}{M}=O(n^r)$ is polynomial in $n$, then Charlie can find an invariant $f'$ of $G$
in polynomial time by solving consecutive linear systems for $d=1, \ldots, \binom{n+M-1}{M}$, each consisting of $m\binom{n+d-1}{d}$ equations in the $\binom{n+d-1}{d}$ variables described in the previous section. For a fixed $d$, this can be accomplished in time $O(m(\binom{n+d-1}{d})^4)$ and the total search will take no more than time
$O(n^{8r})$.
Therefore, for the security of the system it must be guaranteed that $m\binom{n+M-1}{M}$ is not polynomial in $n$.

\subsection{Finding a polynomial invariant of $G$}

We will now discuss an algorithm that will enable us to find an invariant $f'$ of $G$ (and to break the cryptosystem if char $F$ is zero).
The algorithm works inductively, and as a special case, it works when $G$ is a finite group. We will apply this algorithm when char $F$ is zero but it works even when the characteristic of $F$ is finite.

Assume that $H$ is a subgroup of $G$ of finite index in $G$. Assuming we know a nonzero invariant $f$ of $H$, we will find a nonzero invariant of $G$.

\begin{lm}
Let $H$ be a subgroup of $G$ of finite index $s$ in $G$ such that $f$ is an invariant of $H$ of degree $t$.
Then $G$ has a nonzero invariant of degree not exceeding $sM_H$ that can be found in time $O(sn^{t+2}\binom{n+t-1}{t}^{s+1})$.
\end{lm}
\begin{proof}
Denote by $g_1, \ldots, g_s$, where $s=[G:H]$, representatives of all coset classes of $G/H$. Let $f$ be an invariant of $H$ of degree $M_H$.
Denote $x_i=g_if$ for $i=1, \ldots s$, and denote by $p_s(x_1, \ldots, x_s)=x_1\ldots x_s$ the $s$-th elementary symmetric function in  $x_1, \ldots, x_s$.
It is easy to see that $p_s(x_1, \ldots, x_s)$ is invariant with respect to $G$, because each element $g\in G$ permutes coset classes of $G/H$, hence it permutes the set of polynomials $\{x_1, \ldots, x_s\}$. Also, the polynomial $p_s(x_1, \ldots, x_s)=x_1 \ldots x_s$ is nonzero and has the degree $sM_H$.
We can evaluate all polynomials $x_i$ in time $O(sn^2\binom{n+t-1}{t}n^t)$. The product of all $x_i$ can be computed in time $O(\binom{n+t-1}{t}^s)$.
\end{proof}

\begin{cor}
If $G$ is a group of finite order $s$, then the algorithm in the proof of the previous lemma (applied to $H=1$) produces a nonzero invariant of $G$ of order not exceeding $s$ which can be computed in time $O(sn^3n^{s+1})$.
\end{cor}

Note that the time required to run the computation is exponential in the order of $G$ if no invariant of a subgroup of $G$ is known and when we attempt to find an invariant of $G$ from $H=1$.
Nevertheless, there are cases when an invariant of $H$ can be computed in polynomial time; see the next lemma.

The following lemma is well-known, see \cite{burn}.

\begin{lm}\label{burn}
If $G\subset GL_n(\mathbb{R})$ and $G$ is finite, then $G$ has an invariant of degree two.
\end{lm}
\begin{proof}
Let $g_1=1, \ldots, g_s$ be all elements of $G$ and $\mathbb{R}[V]=\mathbb{R}[t_1, \ldots, t_n]$. Denote by $x_i=g_i(t_1^2+\ldots +t_n^2)$ for $i=1, \ldots, s$.  Since values of
each $x_i$ are non-negative when evaluated as polynomials in $t_1, \ldots, t_n$, the values of the invariant polynomial $\sum_{i=1}^s x_i$
evaluated as polynomial in $t_1, \ldots, t_n$ are non-negative and they can be equal to zero only if each
$x_i$ is zero. But $x_1=0$ only if $t_1=\ldots =t_n=0$. Therefore $\sum_{i=1}^s x_i$ is a positive-definite quadratic form in $t_1, \ldots, t_n$,
hence a non-zero invariant of $G$.
\end{proof}

It follows from the previous section that a quadratic invariant of the group $H$, within a context of our public-key cryptosystem, can be found using linear algebra techniques in the polynomial time in $n$. Therefore, for the security of the cryptosystem, we need to make sure that if $H$ is finite, then it is not represented by matrices with real coefficients.

\subsection{Lower bounds for degrees of polynomial invariants}

The significance of understanding the minimal degree $M_{G,V}$ of invariants for the security of the invariant-based cryptosystem was established above. In particular, it is important to find a nontrivial lower bound for $M_{G,V}$.  Unfortunately, we are not aware of any articles establishing lower bounds for the minimal degree of invariants, except in very special circumstances, e.g. \cite{huf}.

On the other hand, there are numerous upper bounds for the minimal degree $\beta(G,V)$ such that $F[V]^G$ is generated as an algebra by all invariants in degrees
not exceeding $\beta(G,V)$.
For example, a classical result of Noether \cite{noeth} states that if the characteristic of $F$ is zero and $G$ is finite of order $|G|$, then $\beta(G,V)\leq |G|$.
There is an extensive discussion of Noether bound and results about $\beta(G,V)$ in section 3 of \cite{smith}.
It was conjectured by Kemper that for $G\neq 1$, and arbitratry ground field $F$, the number $\beta(G, V)$ is at most $\dim V(|G|-1)$. Recently, this conjecture was proved by Symonds in \cite{sym}.

When one wants to find an invariant of $G$, it seems natural to consider an upper bound $\beta(G,V)$. However, if we wants to show that there are no invariants of small degrees
(as is our case), then we need to find lower bounds for $M_{G,V}$. Until now, there was no real impetus to consider such problem. 
We have investigated minimal degree of invariants of $G$ in general in \cite{jmz2}, where we have obtained its description for certain groups $G$.

\section{Cryptosystems based on invariants of infinite diagonalizable groups}\label{zero}

In this section we assume that the characteristic of the ground field $F$ is zero and we design and ivestigate the properties of a cryptosystem based on invariants of an infinite diagonalizable group.

Let us fix a number field $F=\mathbb{Q}(\theta)$ and the subring $Z=\mathbb{Z}[\theta]$ of the ring of algebraic integers of $\mathbb{Q}(\theta)$.
Choose a finite set $Q$ of integers of cardinality $q$ and a set $S_m=\{p_1, \ldots, p_m\}$ of elements of $Z$. 
The elements $p_i$ of $S_m$ need not be primitive and could be units of $Z$. 
Denote by $P_m$ the set of all products of elements from the set $S_m$.

\subsection{Design of the cryptosystems}\label{design}

To start, Alice chooses sets $Q$ and $S_m$ as above.
Afterwards, she chooses her secret key, which is the $n$-tuple of nonnegative integers $(e_1, \ldots, e_n)$, where one component, say $e_n$ equals $1$.
Then she will construct a set of generators $t_1, \ldots, t_s$ of $T$ in such a way that the monomial $f=x_1^{e_1}\ldots x_n^{e_n}$ is invariant under the action of each $t_i$, hence belongs to $F[V]^T$.

At the $i$-th step of the process, Alice chooses the $i$-th generator $t_i$ of the group $T$ as follows.

First, for every $k=1, \ldots, m$ and $1\leq j\leq n-1$ she chooses numbers $b^{(i)}_{k,j}$ from the set $Q$.
Then she computes the numbers $a^{(i)}_1, \ldots, a^{(i)}_{n-1}$ from the set $P_m$ as
$a^{(i)}_j=\prod_{k=1}^m p_k^{b^{(i)}_{k,j}}$ where $1\leq j\leq n-1$.
Alice then computes $a^{(i)}_n$ in such a way that the diagonal matrix $g_i=diag(a^{(i)}_1, \ldots, a^{(i)}_n)$ has $f=x_1^{e_1}\ldots x_n^{e_n}$ as an invariant.
Since she has chosen $e_n=1$, it is easy to see that the appropriate value of $a^{(i)}_n$ is $a^{(i)}_n=\prod_{j=1}^{n-1} (a^{(i)}_j)^{-e_j}.$

Once all generators $t_i$ of the group $T$ are constructed,  Alice chooses an invertible $n\times n$ matrix $P$ as a part of her secret key and computes conjugates $g_i=Pt_iP^{-1}$.
Alice then announces the diagonalizable group $G$ given by its generators $g_i$ for $i=1, \ldots, s$.

When she receives the encrypted message, she can use her secret key $P$ to switch from $G$ to $T$ and apply her previously chosen invariant $f=x_1^{e_1}\ldots x_n^{e_n}$ of $T$
to decrypt the message, as explained in section \ref{system}. She knows that $f$ is an invariant of $T$ because $T$ was constructed to satisfy that condition.

To remove the randomnes of the choice made during this process, Alice should use a cryptographically secure pseudorandom number generator.

\subsection{How to break the cryptosystems in partial cases}
We will explain how the above cryptosystem could be broken in the polynomial time for some rings $Z$.

\begin{lm}\label{break}
Assume that a ring $Z$ is such that the group of units of $Z$ is finite, $Z$ is an Euclidean domain, and the Euclidean algorithm over $Z$ runs in polynomial time in its input.
If a vector, encrypted by the above cryptosystem, has no zero components, then it can be decrypted in polynomial time.
\end{lm}
\begin{proof}
If the group of units of $Z$ is finite, then it consists of roots of unity. Assume that it is generated by $\zeta_E$.

At first, we compute the characteristic polynomials of all matrices $g_i$ and find all of their eigenvalues. This can be done in polynomial time using the algorithm for factoring a polynomial over a number field described in Section 3.6.2. of \cite{cohen}. Hence the factorization of all characteristic polynomials can be done in polynomial time in $n$. 
Then, we follow the algorithm explained in the proof of Proposition 15.4 of \cite{hum} and simultaneously diagonalize all matrices $g_i$
and obtain generators $t_i'$ of the conjugate group $T'$ consisting of diagonal matrices. For simplicity of notation, we can assume that $T'=T$.
Actually, what is important for us are only the eigenvalues $a^{(i)}_1, \ldots, a^{(i)}_{n}$ and their order with respect to fixed order of the eigenvectors. We will not work with
the actual eigenvectors of $V$.

Since all eigenvalues are ratios of elements from $Z$, we can consider the set $X$ of integers that appear in the numerators or denominators of any eigenvalue of any matrix $t_i$.
Using Euclidean algorithm we can compute the set $Y$ of all greatest common divisors of all pairs of elements from $X$. Then we can write a partial factorization of
all elements of $X$ in the form where $x=yz$ and $y$ is a product of elements from $Y$ and $Z$ is not divisible by any element from $Y$. Afterward we replace $X$ by a new set $X'$ consisting
of all elements in $Y$ and of elements $z$ from the above factorization. In the next step we replace the set $Y$ by the set $Y'$ consisting of all greatest common divisors of all pairs of elements from $X'$. We continue in the same fashion and after finitely many steps this process will stabilize. Then we arrive at a set $Y^{(d)}$ of numbers that are pairwise coprime divisors of integers from $X$. Let us call elements of $Y^{(d)}$ atoms of $X$ and denote them by $\{a_1, \ldots a_{m'} \}$.
Since the Euclidean algorithm in $Z$ runs in polynomial time in $n$ and there are no more than $qm$ steps of the above process, we find the atoms in polynomial time in $n$.

For every atom $a$, every element $x$ of $X$ is either coprime to $a$, or is written as $x=a^lb$, where $b$ is coprime to $a$.
Every $a^{(i)}_j$ has the atom factorization $a^{(i)}_j = \zeta_E^{e_{i,j}} \prod_{k=1}^{m'} a_k^{b'^{(i)}_{k,j}}$, where $0\leq e_{i,j} <E$.
We can find an invariant of $T$ from the structure of these diagonal matrices by solving, in nonnegative integers, the system of $s$ equations
$\prod_{j=1}^m (a^{(i)}_j)^{y_j} =1$ in $n$ variables.
Since Charlie has the atom factorization of each element $a^{(i)}_j$, he can compare the exponents in the atom factorization and obtain a system of $s(m'+1)$ linear equations
$\sum_{j=1}^{m'} b'^{(i)}_{k,j} y_j =0$ and $\sum_{j=1}^{m'} e_{i,j} y_j =0$ with bounded coefficients $b'^{(i)}_{k,j}$ and $e_{i,j}$ in $n$ variables $y_j$.
An integer solution of this system can be found in polynomial time in $n$ - see subsection 1.5.2 of \cite{ip}.

The task to find a nonnegative integer solution of the linear system with integer coefficients is an NP-complete problem.
If a solution of our system that has all nonnegative components is found, then it corresponds to a polynomial invariant of $G$.

However, every integral solution corresponds to a rational invariant, that is a rational function that is invariant under the $G$-action.
If the intercepted encoded vector has no zero coordinates, then Charlie can use his rational invariants to decode the message.
\end{proof}

Let us remark that the assumption of the above lemma are satisfied for integers $Z=\mathbb{Z}$ or Gaussian integers $Z=\mathbb{Z}[i]$.
It is well known that the Euclidean algorithm runs in polynomial time over $\mathbb{Z}$ and $\mathbb{Z}[i]$. For a survey of algorithmic results see Section 3 of
\cite{af}.
Also, the above results can be extended further if we replace the assumption that $Z$ is Euclidean domain by the assumption that $Z$ is complex quadratic unique factorization domain.
According to \cite{kr} there is an algorithm, running in polynomial time, that computes gcd in such rings $Z$.

\subsection{Theory of divisibility and units in algebraic number fields}

Assume that $F$ is a number field, that is a finite extension of the field $\mathbb{Q}$ of rational numbers. Let $Z$ be a ring of algebraic integers of $F$.
In many rings $Z$ factorization of elements into primes is not unique (hence $Z$ are not unique factorization domains), see for example the case when $F=\mathbb{Q}(\sqrt{-5})$ in 2.3 of \cite{borsaf}.

Recall the theory of divisibility essentially due to Kummer from Section 3 of \cite{borsaf}. We replace an element $a\in Z$ by the principal ideal $(a)$ of $Z$ generated by $a$.
The ring $Z$ is a Dedekind domain and it has a theory of divisors. I particular, each ideal $\mathfrak{a}$ of $Z$ can be written uniquely as a product
$\mathfrak{a}=\mathfrak{p}_1\ldots \mathfrak{p_r}$, where $\mathfrak{p}_i$ for $i=1, \ldots r$ are prime ideals of $Z$.
The ideals $\mathfrak{a}$ and $\mathfrak{b}$ of $Z$ are called equivalent if there exists $\alpha\in Z^{\times}$ such that $\mathfrak{a}=\mathfrak{b}(\alpha)$. The set of all equivalency classes is called the class group of $F$. It is a finite group of order $h$, called the class number of $F$.
The ring $Z$ is a unique factorization domain if and only if $h=1$.

In relation to factorization of elements from $Z$ it is important to recall the structure of the group $U$ of units of $Z$. By the Dirichlet theorem, 
if $r$ is the number of real embeddings of $F$ and $s$ is the number of complex embeddings of $F$, then the group is isomorphic to a product of a finite group of roots of unity and a free group of rank $r+s-1$, whose generators are called fundamental units of $Z$.

\subsection{Security issue - the choice of the ring $Z$}

The choice of the ring $Z$ is perhaps the most critical since the security of the cryptosystem depends heavily on the arithmetic of the ring $Z$.

The atom or prime factorisation analogous to the one considered in the proof of Lemma \ref{break} is not available in suitable form for number fields in general.
For simplicity assume that $Z=\mathbb{Z}[\theta]$ coincides with the ring of algebraic integers of the field $\mathbb{Q}(\theta)$.
If $Z$ is a principal ideal domain but not a Euclidean domain, then we have a factorization of every element of $Z$ into a product of primitive elements and units of $Z$.
However, without the Euclidean algorithm, it is not clear if we can produce prime factorization of principal ideal in polynomial time.
If $Z$ is not a principal ideal domain, then instead of primitive elements we need to work with divisors. Namely,
for each $x\in Z$ there is the prime ideal decomposition $(x)=\mathfrak{p}_1\ldots \mathfrak{p}_l$, where $\mathfrak{p}_i$ are (not necessarily principal) ideals in $Z$. An ideal generated by each prime number $p$ splits up to a product of many prime ideals (their number does not exceed the degree of the extension $[\mathbb{Q}(\theta):\mathbb{Q}]$, and this number is attained for totally ramified primes $p$). The problem of finding the prime ideal factorization in $Z$ is very difficult. Its special case for $F=\mathbb{Q}$ is the prime factorization problem in $\mathbb{Z}$. The difficulty of factoring of a product of two large primes is the basis of the RSA public-key cryptosystem.
We should consider only those rings $Z=\mathbb{Z}[\theta]$ for which their class number is bigger than one. Such ring $Z$ is not a unique factorization domain (and consequently not a principal ideal domain and not an Euclidean domain).

Even if we assume that the prime factorization of principal ideals generated by $a^{(i)}_j$ is known, by itself it would not be enough to break the above system. The additional
difficulty lies in the structure of the group of units of $Z$. For example, if we choose all elements of $S_m$ to be units of $Z$, then the whole idea of atom or prime decomposition is utterly useless. In order to facilitate the conversion into a system of linear equations, we would need to determine a factorization of each appearing unit into a product of roots or unity and fundamental units of the ring $Z$. Finding a set of fundamental units of the ring $Z$ and decomposition of units of $Z$ into products of root of unity and fundamental units
is by itself a very difficult problem and we are not aware of any algorithm solving these problems in polynomial time.
Therefore the break described in Lemma \ref{break} cannot be duplicated for rings $Z$ that are not unique factorization domains or those containing units of infinite orders.
We remark that there is a plethora of examples of such rings $Z$ appearing in the algebraic number theory.

A combination of obstacles related to factorization of principal ideals and factorization of units of $Z$ as a product of fundamental units is the reason why we propose the above cryptosystem based on approprately selected $Z$. We are unable to find a polynomial algorithm for finding an invariant of the corresponding diagonalizable group $G$.

To summarize, we need to choose the ring $Z$ in such a way that it is not a unique factorization domain and preferably such that its class number is high.
There are numerous examples of rings of integers of number fields that satisfy this condition.
Secondly, we should choose $Z$ so that the rank of its group of units is high. Using Dirichlet theorem, this condition is easy to satisfy.

\subsection{Other security issues}\label{si}

We will consider other possible choices Alice can make and how they affect the security of the system.
The additional choices that affect the security of the crytosystem (besides the choice of $Z$) are the following.

\subsubsection{The choice of the set $S_m$}

We could choose elements $p_i$ from the set $S_m$ in such a way that some of them are primitive.
Also, we should choose them in such a way their norms will have many common prime factors $p$. If we chose them randomly, then there is a great probability that the prime ideals dividing $p$ in the prime decomposition of different $p_i$ are actually different.
Also, we could choose some elements of $S_m$ to be units of $Z$ in order to involve the structure of units of $Z$.

\subsubsection{Choice of the set $Q$}

A choice of a finite set $Q$ does not seem to be important hence we can take it to be small, for example $Q=\{-1, 0,1\}$.

\subsubsection{The choice of the secret key $(e_1, \ldots, e_n)$.}

Another important requirement we need to impose is that none of the entries $(e_1, \ldots, e_n)$ vanishes. The reason for this is to guarantee that the invariant $f$ we chose depends 
on all the variables $x_1, \ldots, x_n$. While we cannot guarantee that there are no invariants of $G$ built on fewer than $n$ variables, chosing our invariant $f$ that depends on all variables is a reasonable precaution. When we increase the number of generators $t_i$ of $T$, it is more likely that $T$ would not have invariants depending on small number of variables. A more careful analysis of this relationship would be desirable.

In order to prevent linear algebra attacks described in Section \ref{LA}, the secret key $(e_1, \ldots, e_n)$ must be chosen so that $E=\sum_{i=1}^n e_i$ is at least of the order of $n$.
For example she can choose $e_i\in \{1,2\}$ such that $\sum_{i=1}^n e_i = [\frac{3n}{2}]$. See also \ref{pa} a) below.

\subsubsection{Choice of the exponents $b^{(i)}_{k,j}$}

We would like to make sure that the minimal degree $M_T$ is close to $E$, which is the degree of $f$, or at least of the order $n$.
However, if the number $s$ of generators $t$ is high and all exponents $b^{(i)}_{k,j}$ are chosen randomly, we expect that $M_T$ is going to be of order $n$.
It is an interesting problem to investigate how to choose $b^{(i)}_{k,j}$ to guarantee that $M_T$ is sufficiently large, say bigger than $E/2$.

If we cannot gurantee that $M_T$ is of order $n$, then we can add another generator $diag(\zeta_E, \ldots, \zeta_E)$ to $T$.
That would require replacing the field $\mathbb{Q}(\theta)$ by $\mathbb{Q}(\theta, \zeta_E)$ and chainging the ring $Z$.

This would give away to Charlie the degree of our invariant $f$ but it would also make sure that $M_T=E$.
Since $E$ is of order $n$, this prevents the linear algebra break discussed in subsection \ref{LA}.

\subsubsection{The choice of the transition matrix $P$.}

The idea of using conjugate group $G$ instead of $T$ is to make matrices representing elements $g\in G$ as far away from the diagonal matrices as possible. Therefore the matrix
$P$ should be complicated, and with many nonzero entries, in order to accomplish this. Please see the next subsection \ref{pa} part b) about the security of conjugation by $P$. 

\subsection{Possible attacks}\label{pa}
We will now describe possible attacks on the above cryptosystem.

{\it a) Linear algebra attack}

Charlie might attempt to find an invariant of $G$ directly using the linear algebra attack described in Section \ref{LA}. The complexity of this approach is exponential
in $n$ if $M_G$ is of the order of $n$, which is likely going to be the case due to (random) choices of $a^{(i)}_j$ and which can be guaranteed by adding
another generator $diag(\zeta_n, \ldots, \zeta_n)$ to $T$. Therefore this linear algebra attack is ineffective.

{\it b) Finding the conjugate group $T$}

Charlie might attempt to find a conjugate group $T'$ of $G$, consisting of diagonal matrices.
In order to diagonalize $G$, he would find all eigenvalues of elements $g_j$ by computing their characteristic polynomials, which he can do in polynomial time in $n$. There exists 
a polynomial algorithm for factoring a polynomial over a number field - it is described in Section 3.6.2. of \cite{cohen}.
Hence the factorization of all characteristic polynomials can be done in polynomial time in $n$. 
Once the eigenvalues of matrices corresponding to every $g_i$ are computed, he can simultaneously diagonalize all matrices $g_i$ (see the proof of Proposition 15.4  of \cite{hum})
and obtain the generators of a group $T'$, in polynomial time in $n$.

This suggests that the conjugation by $P$, suggested in \cite{dima1} as a way of "hiding" the group $G$ and its invariants, is not secure without our context.

{\it c) Finding an invariant using ideal and units factorisation}

For rings $Z$, that are non-Euclidean or have infinite group of units, the attack described in Lemma \ref{break} is not viable.

\subsection{Possible modification of the system}

We have seen before that switching from the system of equations $\prod_{j=1}^m (a^{(i)}_j)^{y_j}=1$ to the linear system $\sum_{j=1}^m b^{(i)}_{k,j} y_j =0$ is important for possible breaking of the system. This can be accomplished by prime ideal factorization - see \ref{pa} c) above.
One possibility to prevent this method is to choose the numbers $a^{(i)}_j$ to be arbitrary and random complex numbers. Then the corresponding linear system would consists of
equations $\sum_{j=1}^m \log(a^{(i)}_j) y_j =0$. It appears to be difficult to find a solution of such general system in integers.

On the other hand for computational purposes we need to approximate the numbers $a^{(i)}_j$ by complex numbers with finite decimal expansions. This would create difficulty estimating
errors of the encryption process. For such system it would be necessary to estimate possible error of encryption and also it would be necessary that the vectors $v_i$ from the set $S$used in the encryption process could be distinguishable within the errors of such computations.

\subsection{More general systems considered in \cite{dima2}}

The main reason we were able to design a system for diagonalizable groups was that we were able to easily construct matrices that have a given monomial as its invariants.
In the case of finite diagonalizable $G$, a reasonable description of the invariants for diagonal matrices is given in \cite{jmz2}. For infinite diagonalizable
$G$ the situation is similar but we have equations instead of congruences.

One could hope that designing a system based on nonabelian $G$ would be more secure than that based on a diagonalizable group $T$ because it is more complicated to
find invariants of such $G$ than those of $T$.
A system based on invariants of nonabelian group $G$ would have an advantage that simultaneous diagonalization as described in \ref{pa} b) is not possible. Therefore the conjugation problem is more difficult to solve for nonabelian $G$.
Also, we need to take into account that the minimal degree of $G$ must be at least of order $n$ to prevent linear algebra attacks.

In the paper \cite{dima2} the authors have proposed a process of generating a more complicated (nonabelian) group $G$, its representation and a corresponding invariant
starting from simpler groups using four types of operations. Their main idea was that it would be more difficult to find an invariant of $G$ than that of the simpler groups.
We will investigate how this construction affects the minimal degrees of invariants since they are
important in regard to the possible linear algebra attack on the corresponding cryptosystem described in subsection \ref{LA}.

For the first operation, assume that $G\leq GL(V)$, where $V\simeq R^n$ is a free module over a ring $R$ of rank $n$; and a ring homomorphism $\pi : R\to R'$,
replacing $R$ with a new ring $R'$, are given.
If $R'$ is a direct summand of $R$ and $\pi$ is a projection onto $R'$ (in which case $R'$ is called smaller), then every invariant of $R[V]^G$ remains an invariant
of $R'[\pi(V)]^G$, hence this operation does not increase the minimal degree $M_{G, V}$.
If $R$ is embedded into $R'$, then $R[V]^G\subseteq R'[R'\otimes_R V]^G$, hence $M_{G, R'\otimes_R V}\leq M_{G, V}$.
The authors of \cite{dima2} do not specify what they mean when $R'$ is larger, and we were unable to follow their arguments.
However, if the kernel of the map $\pi$ is nontrivial, then some of the invariants can be annihilated using this process and the minimal degree can potentially increase.

The second operation replaces $G$ by a conjugated subgroup $H=h^{-1}Gh$ for some $h\in GL(V)$.
Since the algebras/rings $R[V]^G$ and $R[V]^H$ are isomorphic, we have the equality of the minimal degrees $M_{G, V}=M_{H, V}$.

The third operation requires two groups $G_1\leq GL(V_1)$ and $G_2\leq GL(V_2)$ and replaces them by their direct product $G_1\times G_2$ embedded in a natural way into
$GL(V_1\oplus V_2)$. In this case the isomorhism $R[V_1\oplus V_2]^{G_1\times G_2}\simeq R[V_1]^{G_1}\otimes R[V_2]^{G_2}$ implies
$M_{G_1\times G_2, V_1\oplus V_2}=\min\{M_{G_1, V_1}, M_{G_2, V_2}\}$, thus the minimal degree will not increase.

Finally, the fourth operation replaces $G$ by the wreath product $L=G\wr H$, where $H$ is a subgroup of the symmetric group $S_m$.
The group $L$ can be identified with the set of all $m+1$-tuples $(g_1, \ldots , g_m, \sigma)$, where $g_1,\ldots , g_m\in G$ and $\sigma\in H$.
The above element of $L$ acts on $V^{\oplus m}$ by the rule
\[(g_1, \ldots , g_m, \sigma)(v_1, \ldots, v_m)=(g_1 v_{\sigma(1)}, \ldots, g_m v_{\sigma(m)}).\]
The subgroup consisting of all elements with $\sigma=1$ is normal and it is isomorphic to the direct product $G^m$ and
$L$ is isomorphic to the semi-direct product $H\ltimes G^m$. Then $R[V^{\oplus m}]^L=(R[V^{\oplus m}]^{G^m})^H=((R[V]^G)^{\otimes m})^H$ and, for
any invariants $f_1, \ldots, f_m$ from $R[V]^G$, the element
\[\sum 1\otimes\ldots\otimes \underbrace{f_i}_{i-\mbox{th place}}\otimes\ldots\otimes 1\] is $L$-invariant.
Therefore $M_{L, V^{\oplus m}}\leq M_{G, V}$.

Summing up, all four operations as presented in \cite{dima2} do not increase the minimal degrees of invariants of given representations of the initial groups
(possibly with the exception of the first operation with non-injective map $\pi$).
Therefore, regardless of how complicated the resulting group $G$ and its representation is, it is no more secured against the linear algebra attack described in subsection \ref{LA}
and great care needs to be taken that the initial minimal degrees of the starting groups are large enough, say of the order $n$.
On the other hand, if the minimal degrees of the starting group is sufficiently large, then from the point of view of such linear algebra attack
it is not necessary to construct a more (structurally) complicated group or representation.

\subsection{Invariants of supergroups}

Another possible modification of the cryptosystem is obtained when the group $G$ and its invariants are replaced by a supergroup and its superinvariants. 
A significant difference that is exhibited in this case is that invariants of supergroups do not have a basis consting of monomials. Thus the structure of the invariants of supergroups is more complicated. 

We will not go further into rather complicated details about supergroups and their invariants but would like to refer an interested reader to the paper \cite{jmz2} where we have obtained results in this direction and stated potential application in cryptosystems based on relative invariants and absolute invariants of supergroups.

\section{Cryptosystem over finite rings $R$}

In order to make the cryptosystem build in Section \ref{zero} more effective and easier to implement we will make modification that will work over finite rings $R$ instead of over fields $F$.

The motivating example is the residue class ring $R$ of the ring of algebraic fields modulo an integer $m>1$. However we can consider cryptosystem over arbitrary finite commutative ring $R$.

\subsection{Structure theory for finite commutative rings $R$}

Recall the following structure theorems for finite commutative, finite local commutative rings and their groups of units from \cite{mcd}.

\begin{teo}(Theorem VI.2 of \cite{mcd})
Let $R$ be a finite commutative ring. Then $R$ is isomorphic to a direct sum of local rings.
\end{teo}

\begin{teo}(Theorems XVII.1 and XVIII.2 of \cite{mcd})
Let $R$ be a finite local commutative ring of characteristic $p^n$ with maximal ideal $\mathfrak{m}$ and residue field $k$.
Let $[k:\mathbb{Z}_p]=r$ and $u_1, \ldots, u_t$ be a minimal $R$-generating set of $\mathfrak{m}$.
Then the largest Galois extension $T$ (called the coefficient ring of $R$) of $\mathbb{Z}_p$ in $R$ is isomorphic to the Galois ring $GR(p^n, r)$,
and $R$ is a ring homomorphic image of the polynomial ring $T[X_1, \ldots, X_t]$.

The group of units $R^{\times}$ of $R$ is isomorphic to $(1+\mathfrak{m}) \times k^{\times}$. The Abelian $p$-group $1+\mathfrak{m}$ is called the one group of $R$.
\end{teo}

Gilmer has characterized when $R^{\times}$ is cyclic.

\begin{teo}(Theorem XVIII.9 of \cite{mcd})
Let $R$ be a finite local commutative ring. Then $R^{\times}$ is cyclic if and only if $1+\mathfrak{m}$ is cyclic. In this case $R$ is isomorphic to one of the following
\begin{itemize}
\item{$GF(p^t)$ (if $\mathfrak{m}=0$)}
\item{$\mathbb{Z}/p^n\mathbb{Z}$ (if $p\geq 3$ and $n>1$)}
\item{$\mathbb{Z}/4\mathbb{Z}$}
\item{$(\mathbb{Z}/p\mathbb{Z})[X]/(X^2)$}
\item{$(\mathbb{Z}/2\mathbb{Z})[X]/(X^3)$}
\item{$(\mathbb{Z}/4\mathbb{Z})[X]/(2X,X^2-2)$}
\end{itemize}
\end{teo}

Hence in most cases $R^{\times}$ is not cyclic.

We will work with finite rings $R$ obtained as residue rings of the ring of algebraic integers $Z$ of number fields $F$ modulo an ideal $\mathfrak{a}$ of $Z$; and specialize further to 
the case when $\mathfrak{a}=(m)$ and $m\in Z$.

The structure of units $U(R)$ of the residue classes $R$ of $Z$ modulo a power of a prime ideal $\mathfrak{p}$ is rather complicated and described in Theorem 2 and 3 of \cite{nak}. For simplicity we quote here only Theorem 1 of \cite{nak} describing the $p$-rank of $U(Z/\mathfrak{p}^{N+1})$.

Let $\mathfrak{p}$ be a prime ideal of $Z$ dividing a prime number $p$ of $\mathbb{Q}$, $e$ and $f$ be the ramification index and the degree of $\mathfrak{p}$ over $\mathbb{Q}$. Denote
by $e_1=[\frac{e}{p-1}]$. The following statement is due to Hasse and Takenouchi.

\begin{teo}(Theorem 1 of \cite{nak})  
The $p$-rank $R_N$ of $U(Z/\mathfrak{p}^{N+1})$ is given by
\begin{itemize}
\item{$(N-[\frac{N}{p}])f$ if $0\leq N<e+e_1$}
\item{$ef$ if $N\geq e+e_1$ and the primitive $p$-th root of unity does not belong to the $p$-adic completion $F_{\mathfrak{p}}$}
\item{$ef+1$ if $N\geq e+e_1$ and the primitive $p$-th root of unity belongs to the $p$-adic completion $F_{\mathfrak{p}}$.}
\end{itemize}
\end{teo}

An interesting connection of the structure $U(Z/\mathfrak{a})$ to multiplicative semigroups and Fermat-Euler theorem in algebraic number field is given in \cite{laspor}.

\subsection{Modification of the cryptosystem to finite rings}

Assume that $F$ is a number field, that is a finite extension of the field $\mathbb{Q}$ of rational numbers. Let $Z$ be a ring of algebraic integers of $F$.

Let $\mathfrak{a}=\mathfrak{p}_1^{k_1}\ldots \mathfrak{p}_r^{k_r}$, where $\mathfrak{p}_i$ for $i=1, \ldots r$ are distinct prime ideals of $Z$ and $k_i$ are the corresponding multiplicities. We say that elements $\alpha, \beta \in Z$ are called congruent modulo $\mathfrak{a}$, and write $\alpha\equiv \beta \pmod{\mathfrak{a}}$ if $\alpha-\beta$ is divisible by $\mathfrak{a}$. The equivalence classes of this congruence form a residue class ring $R=Z/\mathfrak{a}$. Then $R$ is a finite commutative ring that is isomorphic to a direct sum of local rings $Z/\mathfrak{p}_i^{k_i}$. If $\mathfrak{p}$ has the degree of inertia $f$, then the norm $Norm(\mathfrak{p})=p^f$ for a rational prime $p$ and the cardinality of the local ring $Z/\mathfrak{p}^{k}$ is $p^{fk}$. In particular, the residue class ring $Z/\mathfrak{p}$ is a finite field of cardinality $q=p^f$.
This way we have a concrete realization of the structure theorem for finite commutative rings for residue class rings $R=Z/\mathfrak{a}$.

There is a theory of divisors on $R$ induced from the theory of divisors on $Z$. The finite ring $R$ has divisors of zero and the ideals of $R$ are in one-to-one correspondence to ideals of $Z$ dividing $\mathfrak{a}$. Thus $R$ has $r$ distinct prime ideals $\overline{\mathfrak{p}_i}$ for $i=1, \ldots, r$ and the product
$\overline{\mathfrak{p}}_1^{k_1}\ldots \overline{\mathfrak{p_r}}^{k_r}$ vanishes.

\subsubsection{Choice of the ring $R$ and its ideal $\mathfrak{a}$}

Let us make the following observation. 

\begin{pr} Let $Z$ be the ring of algebraic integers of a number field $F$, and $\mathfrak{a}$ be an ideal of $Z$. Then the residue class ring $R=Z/\mathfrak{a}$ is 
a principal ideal ring that is isomorphic to a finite product of local rings. 
\end{pr}
\begin{proof}
Let $\mathfrak{a}=\mathfrak{p}_1^{k_1}\ldots \mathfrak{p}_r^{k_r}$, where $\mathfrak{p}_i$ for $i=1, \ldots r$ are distinct prime ideals of $Z$ and $k_i$ are the corresponding multiplicities.

For each $i=1, \ldots, r$ denote by $Z_i$ the localization of $Z$ with respect to the prime ideal $\mathfrak{p}_i$ and by $R_i=Z_i/\mathfrak{p}_i^{k_i}$ the local ring with the 
maximal ideal $\mathfrak{p}_iR_i$. It follows from the Chinese remainder theorem that the map $\phi:R \to \prod_{i=1}^r R_i$ that sends $ x \pmod{\mathfrak{a}}$ to 
$(x \pmod{\mathfrak{p}_i^{k_i}})_{i=1}^r$ is an isomorphism of rings. Moreover, the groups of units $U(R)$ of $R$ and $U(\prod_{i=1}^r R_i)$ are isomorphic under this map.

Ideals $\overline{\mathfrak{b}}$ of $R$ are of the form $\overline{\mathfrak{p}}_1^{l_1}\ldots \overline{\mathfrak{p}}_r^{k_r}$, where $0\leq l_i\leq k_i$ for each $i=1, \ldots, r$.
If we choose elements $\pi_i$ such that $\pi\in \mathfrak{p}_i\setminus \mathfrak{p}_i^2$ ($\pi_i$ uniform element with respect to the valuation of 
$Z$ corresponding to $\mathfrak{p}_i$) and $\pi_i\equiv 1 \pmod{\mathfrak{p}_j}$ for $j\neq i$ (this is possible by the Chinese remainder theorem), then 
$\overline{\mathfrak{b}}=(\prod_{i=1}^r \overline{\pi}_i^{l_i})$.
\end{proof}

Let us note that the groups of units $U(R)$ of $R$ and $U(\prod_{i=1}^r R_i)=\prod_{i=1}^r U(R_i)$ are also isomorphic under the map $\phi$, and each group $U(R_i)$ is cyclic and isomorphic to the group of units of the finite field $Z/\mathfrak{p}_i$.

It might appear that due to the above Proposition, when we pass from the ring $Z$ that is not a unique factorization domain (when its class number $h>1$) to the residue class ring 
$R=Z/\mathfrak{a}$, the bad properties of factorization in $Z$ do not carry over to $R$ since $R$ is a principal ideal ring.

However, the complexity of the divisor theory of $R$, together with a suitable choice of $\mathfrak{a}$, indeed influences the complexity of the structure of $R$.
If the number $r$ of primary factors $\mathfrak{p}_i^{k_i}$ and/or the exponents $k_i$ are high and the factorization of $\mathfrak{a}$ is not available, it will be difficult to derive the structure of $R$ effectively from the ring $Z$ and its ideal $\mathfrak{a}$. 

Determining the abstract structure of an arbitrary finite ring $R$ seems to be even more complicated problem.
 
Based on this disussion, we should choose the residue ring $R$ and its ideal $\mathfrak{a}$ in such a way that the number $r$ of primary components of $\mathfrak{a}$ is high and the 
prime factorization $\mathfrak{p}_1^{k_1}\ldots \mathfrak{p_r}^{k_r}$ of $\mathfrak{a}$ in $Z$ is complicated.

We will consider only the case when $R=Z/(m)$ for the ring of algebraic integers of a number field $F$ and integer $m>1$. The choice of $F$ and $m$ cannot be independent since even 
if $Z$ has a complicated theory of divisors, choosing wrong $m$ can create $R$ that is rather simple. 

We will now modify previously defined cryptosystem based on invariants of diagonalizable matrices to the case when $R$ is the residue class ring $R=Z/(m)$ where
$Z$ is the ring of algebraic integers of a number field $F$ and $m$ is an integer.

However, we start with $F$ with complicated divisor theory and then look for appropriate 
$m$ so that the structure of the residue ring $Z/(m)$ contains the complication of the factorization in $Z$ and also problem of the factorization of the modulus $m$.

Let $m=p_1^{d_1}\ldots p_s^{d_s}$ and each $p_i$ decomposes as $p_i=\mathfrak{p}_{1,i}^{e_{1,i}}\ldots \mathfrak{p_{s,i}}^{e_{s,i}}$ in $Z$.
Let us consider the following cases for $m$.

1) $m$ is square-free. In particular, one appealing choice is when $m$ is a product of two large primes $p_1$ and $p_2$.
In this case the ring $R$ will be a direct sum of finite fields corresponding to residue class rings of $Z$ by unramified prime divisors of $m$.
We should choose the prime factors of $m$ to be unramified in $F$ because otherwise it would be easier to factor $m$ by considering the greatest common divisor of
$m$ and the discriminant of $Z$. Special case to consider is when either $p_1$, $p_2$ or both are totally unramified.


If we want to involve local rings that are not finite fields in the decomposition of $R$, we have additional choices for $m$.

2) We can choose $m$ that is not square-free. For example, we can choose $m=(p_1p_2)^2$, where both $p_1$ and $p_2$ are large primes. In this case we have the factorization problem for $p_1p_2$ but since $m$ is not square-free, it is much easier to find prime factorization of $m$ when compared to the square-free case. 

If we consider a more general case and replace the integer $m$ by an ideal $\mathfrak{a}$ (principal or not) of $Z$ such that 
$\mathfrak{a} =\mathfrak{p}_1^{k_1}\ldots \mathfrak{p}_r^{k_r}$, then the structure of $U(Z/\mathfrak{a})$ is even more complicated because 
the group of units of the local rings $Z_i/(\mathfrak{p}_{i}^{k_i})$ might have high $p$-rank if the exponents $k_i$ are high.

\subsubsection{Choice of the group $G$} 
Assume that we have already chosen $Z$, $m$ and the corresponding complicated residue ring $R=Z/(m)$. 
Let us modify the cryptosystem introduced in \ref{design} in such a way that instead of working inside the ring $Z$ we will work inside the residue ring $R=Z/(m)$.

We will assume that the entries in the matrices representing generators $g_i$ of $G$ are units of $R$.
Assume that the group of units $U$ of the ring $R$ has a basis given by $u_1, \ldots, u_r$ of respective orders $o_1, \ldots, o_r$,
generators $g_i$ of $G$ are given as 
$g_i=\begin{pmatrix} r_{i1} &0 & \ldots &0 \\0&r_{i2}& \ldots&0\\\ldots &\ldots&\ldots&\ldots\\0&0&\ldots&r_{in}\end{pmatrix}$ for $i=1, \ldots t$,
where $r_{ij}=u_1^{l_{ij1}}\ldots u_r^{l_{ijr}}$ for appropriate exponents $l_{ijk}$, and the set $S$ consists of vectors 
$\vec{v}_i=\begin{pmatrix}a_{i1}\\a_{i2}\\\ldots \\a_{in}\end{pmatrix}$ for $i=1, \ldots, s$.
The general element of $G$ is written as $g=g_1^{y_1}\ldots g_t^{y_t}$ for some integers $y_1, \ldots, y_t$.

A monomial $f=x_1^{d_1}x_2^{d_2}\ldots x_n^{d_n}$ separates $\vec{v}_i$ and $\vec{v}_j$ if and only if 
\[a_{i1}^{d_1}a_{i2}^{d_2}\ldots a_{in}^{d_n}\neq a_{j1}^{d_1}a_{j2}^{d_2}\ldots a_{jn}^{d_n}.\]

Since every invariant of $G$ is a sum of monomial invariants, to obtain a complete description of all invariants of $G$ we only need to find monomial invariants.
A monomial $f=x_1^{d_1}x_2^{d_2}\ldots x_n^{d_n}$ is an invariant of $G$ if 
$\prod_{i=1}^t r_{i1}^{y_id_1}r_{i2}^{y_id_2}\ldots r_{in}^{y_id_n}=1$ for all integers $y_1, \ldots, y_t$. 
This implies 
$r_{i1}^{d_1}r_{i2}^{d_2}\ldots r_{in}^{d_n}=1$ for each $i=1, \ldots, t$.

For fixed $i$, the condition $r_{i1}^{d_1}r_{i2}^{d_2}\ldots r_{in}^{d_n}=1$ gives 
$u_1^{\sum_{j=1}^n d_jl_{ij1}}\ldots u_r^{\sum_{j=1}^n d_jl_{ijr}}=1$
and that is equivalent to the system of congruences
\[d_1l_{i1k}+d_2l_{i2k}+\ldots + d_nl_{ink}\equiv 0 \pmod{o_k}\]
for each $k=1, \ldots, r$.

Since $i$ runs from $1$ to $t$, the condition that $f$ is an invariant of $G$ is equivalent to $t$ such systems. In total we obtain a system of $rt$ congruences in variables
$d_1, \ldots, d_n$.

If we assume that numbers $l_{ijk}$ are determined, then the last system can be solved using integer programming since it is equivalent to 
the system of $rt$ equations 
\[d_1l_{i1k}+d_2l_{i2k}+\ldots + d_nl_{ink}+ o_kd_{ik} =0\]
for each $i=1, \ldots, t$ and $k=1, \ldots r$
in integer variables $d_1, \ldots, d_n$ and $rt$ variables $d_{ik}$.

When we worked over the finite field $F=GF(q)$ we were able to determine the corresponding coefficients $l_{ij}$ using discrete logarithms. In the case of finite rings $R$ we do not have such a tool at our disposal. If we want to determine coefficients $l_{ijk}$, first we need to determine the structrure of the ring $R$ and its group of units. 
Papers \cite{sutherland} and \cite{bjt} provide algorithms that perform this task with the complexity $O(|R|^{\frac12})$.

Even when the basis $u_1, \ldots, u_r$ of $U$ is known it is not clear how to determine the coefficients $l_{ijk}$. This issue can make the use of the cryptosystems over finite rings more appealing.

Important note is that all the difficulties we have just described are needed for breaking of the cryptosystem. However for the design of the cryptosystem we do need to know neither 
a basis $u_1, \ldots, u_r$ of $U$ nor the exponents $l_{ijk}$. All that is required is to check that the diagonal entries are units in the residue ring $R$. 
Since our ring $R$ is the residue class of $Z$ modulo $m$, this condition can be verified by computing the norm of these entries. 
If all the norm of the entry is coprime to $m$, then its image in $R$ belongs to $U$.  

\subsubsection{Conclusion}

The breaking of the modified cryptosystem designed over $R$ seems to require techniques going beyond discrete logarithm problem. More work that is required to specify the parameters of the cryptosystem that provide its sought-after security is beyond the scope of this paper. We hope that we have convinced the reader that this is a worthwhile endeavour to undertake.   

In the papers \cite{dima1} and \cite{dima2} the cryptosystem based on invariants of groups over a field $F$ were considered. The above modification of our cryptosystem works over finite rings instead of fileds. The reason why we need the ring structure is because we want to use the matrix multiplication and conjugation by a matrix $P$ to hide the group $G$ as suggested by Grigoriev. If it is determined that the conjuagation by $P$ does not increase the security of the cryptosystem, then we can consider a more general setup and instead of working over finite rings we could work over finite multiplicative groups. That is another direction for future investigation.


\end{document}